\documentclass[conference,letterpaper]{IEEEtran}

\usepackage[
  top=0.69in,
  bottom=1.04in,
  left=.635in,
  right=.63in,
]{geometry}

\usepackage[utf8]{inputenc}
\usepackage[T1]{fontenc}
\usepackage{url}
\usepackage{ifthen}
\usepackage[cmex10]{amsmath}

\usepackage[ruled]{algorithm}
\usepackage{colortbl}
\usepackage{caption}
\usepackage{float}
\usepackage{multicol}
\usepackage[normalem]{ulem}
\usepackage[noend]{algpseudocode}
\usepackage{tikz-layers}
\algnewcommand{\Initialize}{%
  \State \textbf{Initialize:}
}
\algnewcommand{\Output}{%
  \State \textbf{Output:}
}
\usepackage{mathtools}
\usepackage{url}

\usepackage[
    shortcuts,
    acronym,
    section=section,
]{glossaries}
\usepackage{tikz}
\usetikzlibrary{backgrounds,calc,shadings,shapes.arrows,shapes.symbols,shadows,fit,positioning,topaths,hobby,arrows.meta, decorations.pathmorphing}
\usepackage{float}
\usepackage{amssymb}
\usepackage{amsthm}
\usepackage{xspace}
\usepackage[capitalize]{cleveref}
\usepackage[inline]{enumitem}
\usepackage{cite}
\usepackage{siunitx}
\usepackage{dsfont}
\usepackage{balance}
\makeglossaries

\newtheorem{theorem}{Theorem}

\newtheorem{lemma}{Lemma}
\newtheorem{proposition}{Proposition}
\newtheorem{definition}{Definition}

\newtheorem{example}{Example}
\newtheorem{remark}{Remark}

\usepackage{pgfplots}
\pgfplotsset{compat=newest}

\crefname{paragraph}{paragraph}{paragraphs}
\Crefname{paragraph}{Paragraph}{Paragraphs}

\IEEEoverridecommandlockouts

\usepackage{mleftright}

\usepackage[dvipsnames]{xcolor}

\usepackage{graphicx}
\usepackage{caption}
\usepackage{subcaption}

\pgfplotsset{compat=1.18}

\newcommand{\extversion}{the appendix}

\newcommand{\NN}{\mathbb{N}}

\newcommand{\F}{\mathbb{F}}
\newcommand{\6}{\mathbf}

\newcommand{\vidx}{i}
\newcommand{\size}[1]{\left\lvert#1\right\rvert} 
\newcommand{\Iset}{\mathcal{I}}
\newcommand{\bigO}[1]{O(#1)}
\newcommand{\smallO}[1]{o(#1)}
\newcommand{\getrand}{\xleftarrow{\smash{\raisebox{-0.4ex}{\tiny$\mathsf{R}$}}}}
\newcommand{\Prob}[1]{\textup{Pr}\mleft[#1\mright]}
\newcommand{\Bern}[1]{\mathrm{Bern}(#1)}% Bernoulli
\newcommand{\Dset}{\mathcal{D}}

\newcommand{\LPN}[1][\plpn]{\mathsf{LPN}_{\m,\n}^{#1}}
\newcommand{\LPNattack}[2]{\mathsf{LPN}_{#2}^{#1}}
\newcommand{\plpn}{\tau} % noise
\newcommand{\m}{m} % secret dimension
\newcommand{\n}{n} % samplet dimension
\newcommand{\q}{{q}} % modulo
\newcommand{\Am}{\6A}

\newcommand{\HintLPN}[1][\pe,\pf]{\mathsf{HintLPN}^{#1}_{\m,\n}}
\newcommand{\p}{p}
\newcommand{\pe}{p} % Hint-LPN noise
\newcommand{\pf}{p'} % hint noise
\newcommand{\KAHE}{\mathtt{KAHE}\xspace}
\newcommand{\KGenKAHE}{\mathtt{KeyGen}\xspace}
\newcommand{\EncKAHE}{\mathtt{Enc}\xspace}
\newcommand{\DecKAHE}{\mathtt{Dec}\xspace}
\newcommand{\AddKAHE}{\mathtt{Add}\xspace}
\newcommand{\SetupKAHE}{\mathtt{Setup}\xspace}
\newcommand{\C}{\mathcal{C}}
\newcommand{\ct}{\mathbf{y}}
\newcommand{\pKAHE}{p}
\newcommand{\nC}{{n}}
\newcommand{\kC}{k}
\newcommand{\define}{\triangleq}

\newcommand{\dec}{\mathsf{dec}}
\newcommand{\enc}{\mathsf{enc}}
\newcommand{\cw}{\6c}

\newcommand{\uidx}{j}
\newcommand{\uinput}[1][\uidx]{\mathbf{m}_{#1}}
\newcommand{\uinputagg}[1][\uidx]{\bar{\mathbf{m}}}
\newcommand{\qints}{\Lambda}
\newcommand{\Ncm}{M}

\newcommand{\cmidx}{\ell}

\newcommand{\sskeyidx}{\tau}
\newcommand{\pidx}{\kappa}
\newcommand{\Ncrt}{T}
\newcommand{\crtidx}{t}
\newcommand{\crtmod}[1][\crtidx]{\q_{#1}}
\newcommand{\ukey}[1][\uidx]{\6s_{#1}}

\newcommand{\ukeyagg}[1][\crtidx]{\bar{\6s}}

\newcommand{\ukeypart}[1][\pidx]{\6s_{\uidx, #1}}
\newcommand{\sskey}[1][\sskeyidx]{\mathbf{r}_{\uidx}}
\newcommand{\sskeypart}[1][\pidx]{\mathbf{r}_{\uidx, #1}}
\newcommand{\Z}{\mathbb{Z}}

\newcommand{\PKE}{\mathtt{PKE}\xspace}
\newcommand{\EncPKE}{\PKE.\mathtt{Enc}\xspace}
\newcommand{\DecPKE}{\PKE.\mathtt{Dec}\xspace}
\newcommand{\cmskey}[1][\cmidx]{\mathsf{sk}_{#1}}
\newcommand{\cmpkey}[1][\cmidx]{\mathsf{pk}_{#1}}
\newcommand{\dskey}[1][\cmidx]{\6s\6k}
\newcommand{\dpkey}[1][\cmidx]{\6p\6k}
\newcommand{\ctd}[1][\cmidx]{\mathsf{ct}_{\uidx}}
\newcommand{\ctcm}[1][\cmidx]{\mathsf{ct}_{\uidx, #1}}
\newcommand{\ctcrt}[1][\uidx]{\ct_{#1}}

\newcommand{\ctagg}[1][\crtidx]{\bar{\ct}}

\newcommand{\aggpoly}{F}
\newcommand{\poly}[1][\uidx]{f_{#1}}

\newcommand{\eval}[1][\cmidx]{\alpha_{#1}}
\newcommand{\evalcrt}[1][\cmidx]{\alpha_{#1, \crtidx}}

\newcommand{\share}{\poly(\eval)}
\newcommand{\aggshare}{\aggpoly(\eval)}

\newcommand{\keyfuncinv}[1][\ukey]{f^{-1}}

\newcommand{\colluser}{\mathcal{Z}} % colluding users
\newcommand{\uinputcrt}[1][\uidx]{\mathbf{m}_{#1, \crtidx}}
\newcommand{\uinputcrtagg}[1][\crtidx]{\bar{\mathbf{m}}_{#1}}

\newcommand{\Hyb}[1]{\mathtt{H}^{(#1)}}
\newcommand{\hidx}{r}
\newcommand{\advA}{\mathcal{A}}
\newcommand{\advB}{\mathcal{B}}

\pgfplotsset{
  N10/.style={
    line width=0.9pt,
  },
  N100/.style={
    line width=1pt,
    densely dotted,
  },
  N50/.style={
    line width=0.9pt,
    dashed,
  },
  solver/.style={
    color = NavyBlue,
    line width=0.9pt,
  },
  crt/.style={
    color = RedOrange,
    line width=0.9pt,
  },
  swift/.style={
    color = Gray,
    line width=0.9pt,
  },
  N10solver/.style={
    N10, solver,
  },
  N10crt/.style={
    N10, crt,
  },
  N10swift/.style={
    N10, swift,
  },
  N100solver/.style={
    N100, solver,
  },
  N100crt/.style={
    N100, crt,
  },
  N100nocrt/.style={
    N100, nocrt,
  },
  N100swift/.style={
    N100, swift,
  },
  N50solver/.style={
    N50, solver,
  },
  N50swift/.style={
    N50, swift,
  },
   N50crt/.style={
    N50, crt,
  },
}

\newcommand{\Nuser}{N} % number of parties
\newcommand{\coll}{z} % collusion parameter
\newcommand{\kinfo}{k} % number of information symbols

\colorlet{q2}{NavyBlue}
\colorlet{q7}{BrickRed}
\colorlet{q31}{ForestGreen}
\colorlet{q127}{Purple}

\pgfplotsset{
  norm10/.style={
    line width=0.9pt,
    color = NavyBlue,
  },
  norm50/.style={
    line width=0.9pt,
    color = BrickRed,
  },
  norm100/.style={
    line width=0.9pt,
    color = ForestGreen,
  },
  normSwift/.style={
    line width=0.9pt,
    color = black,
  },
}

\newcommand{\ppile}[1]{\p^{\oplus #1}} % piled-up noise rate

\newcommand{\dd}{\mathrm{d}}

\newcommand{\qentropy}[2][]{%
  \if\relax\detokenize{#1}\relax H_{\q}\else{H_{\q}}#1\fi\mleft(#2\mright)}

\providecommand{\extversion}{\cite{extended}}
\newcommand{\agg}[1]{\bar{#1}}

\newcommand{\Rate}{R}
\newcommand{\mV}{\6m}
\newcommand{\mVu}[1][\uidx]{\6m_{#1}}
\newcommand{\maggV}{\agg{\6m}}
\newcommand{\eV}{\6e}
\newcommand{\eVu}[1][\uidx]{\6e_{#1}}
\newcommand{\eaggV}{\agg{\6e}}
\newcommand{\eaggS}{\agg{e}}
\newcommand{\sV}{\6s}
\newcommand{\sVu}[1][\uidx]{\6s_{#1}}
\newcommand{\saggV}{\agg{\6s}}
\newcommand{\aM}{\6A}

\newcommand{\ctVu}[1][\uidx]{\6y_{#1}}
\newcommand{\ctaggV}{\agg{\6y}}

\newcommand{\capprox}{\approx_{\mathrm{c}}}
\newcommand{\DistA}{\mathcal{D}_{\mathrm{A}}} 

\newcommand{\fV}{\6f}
\newcommand{\uV}{\6u}
\newcommand{\uVu}[1][\uidx]{\6u_{#1}}
\newcommand{\hV}{\6h}
\newcommand{\bV}{\6b}

\newcommand{\cVu}[1][\uidx]{\6c_{#1}}
\newcommand{\postpp}[2]{\tau(#1,#2)} % posterior noise rate
\newcommand{\postp}[1]{\tau(#1)} % posterior noise rate

\newcommand{\fS}{f}
\newcommand{\eS}{e}
\newcommand{\hS}{h}
\newcommand{\tS}{t}
\newcommand{\tV}{\6t}
\newcommand{\xS}{x}

\newcommand{\given}{\,\vert\,}
\newcommand{\extdeg}{d}
\newcommand{\extdegcrt}[1][\crtidx]{\extdeg_{#1}}
\newcommand{\mshare}{w}

\begin{document}
\title{Post-Quantum Secure Aggregation via \\Code-Based Homomorphic Encryption \vspace{-.15cm}}

\author{%
    \IEEEauthorblockN{Sebastian Bitzer, Maximilian Egger, Mumin Liu$^\star$ and Antonia Wachter-Zeh \\
    Technical University of Munich, Munich, Germany,\\ \{sebastian.bitzer, maximilian.egger, antonia.wachter-zeh\}@tum.de, liumumin2@gmail.com \vspace{-.6cm}%
                    }
    \thanks{This project has received funding from Agentur für Innovation in der Cybersicherheit GmbH. $^\star$Work done as part of his Master's Thesis at TUM.}
}

\maketitle

\begin{abstract}
Secure aggregation lets a server learn the sum of inputs from multiple users without learning individual contributions.
Recent single-server protocols achieve this with key- and message-additive homomorphic encryption (KAHE), but existing instantiations rely on lattice assumptions.
We instantiate a KAHE from the Learning Parity with Noise (LPN) assumption and analyze the challenges that arise with the Hamming metric.
Decryption of the aggregate amounts to decoding over a $q$-ary symmetric channel whose crossover probability grows with the number of users.
Security under leakage of the key sum reduces to LPN with a noise rate that is quadratic in the KAHE noise rate, and a matching attack shows that this loss is inherent when all but two users collude.
Asymptotically, we characterize the optimal communication cost in closed form, showing that the overhead over the raw payload vanishes as the input length grows, and decouple the field size from the input range via the Chinese remainder theorem.
In cross-silo settings, our protocol requires less communication than information-theoretically secure aggregation for moderate to high collusion.
\end{abstract}

\section{Introduction}
\textbf{Secure aggregation} allows a server to obtain the sum of users' inputs without learning individual contributions.
It is a key component of federated learning (FL), where model updates of users must be aggregated without compromising their privacy~\cite{bonawitz2017practical}.
All schemes tolerate only a bounded number of colluders, but the way their performance scales differs.
With information-theoretic masking, as in~\cite{zhao2022information,jahani2022swiftagg+,egger2025private,kadhe2020fastsecagg,so2021turbo,so2022lightsecagg}, the randomness must be as long as the input, so the overhead of strong collusion resilience, in communication or shared randomness, is incurred on the entire input.
Cryptographic approaches, such as~\cite{bonawitz2017practical,bell2020secure}, expand short secrets into long masks under computational assumptions and thus incur this overhead only on the short secrets.
Since exchanged ciphertexts may be recorded today and decrypted once large-scale quantum computers exist~\cite{mosca2018cybersecurity}, these assumptions should be post-quantum secure.

\textbf{Hybrid secure aggregation.}
The recent single-server protocols OPA~\cite{C:KarPol25} and Willow~\cite{C:BGLRS25} follow this principle with a \emph{key- and message-additive homomorphic encryption (KAHE)} scheme: each user encrypts its input under a fresh key whose length is independent of the input length, the server adds the ciphertexts, and only the keys are aggregated securely, by a committee or a non-colluding helper.
This is the secure-aggregation analogue of hybrid encryption.
Both protocols instantiate KAHE from Learning with Errors (LWE) variants~\cite{STOC:Regev05}, as does follow-up work~\cite{madathil2025tacita,taiello2025buffalo}.
For crypto-agility and to guard against advances in lattice cryptanalysis, instantiations from other post-quantum assumptions are desirable.

\textbf{Code-based KAHE.}
We instantiate KAHE from the Learning Parity with Noise (LPN) assumption~\cite{C:BFKL93}, which states that decoding a random linear code from errors of low Hamming weight is hard.
We adopt the protocol structure of~\cite{C:KarPol25,C:BGLRS25}; combined with a code-based mechanism for key transport, such as HQC~\cite{NISTPQC-R4:HQC22}, this yields a fully code-based protocol.
Moving from the Euclidean metric of LWE to the Hamming metric is not straightforward: in the Euclidean metric, noise grows gracefully under addition, and tools such as Hint-LWE~\cite{C:KLSS23b}, on which the security of~\cite{C:BGLRS25} rests, have no established LPN analogue.
Our work addresses these challenges:
\begin{enumerate}[label=(\roman*)]
\item \emph{Correctness.}
Decrypting the aggregate requires decoding in the presence of the sum of the users' errors.
For the LWE-based $\KAHE$ of~\cite{C:BGLRS25}, this only requires a sufficiently large ciphertext modulus, at a cost logarithmic in the number of users $\Nuser$.
For our LPN-based instantiation, it amounts to decoding over a $\q$-ary symmetric channel whose crossover probability grows roughly linearly in $\Nuser$.
\item \emph{Security: Hint-LPN reduction.}
To decrypt, the server learns the key sum, which reveals a noisy version of each user's error.
Under this leakage, security follows from a Hint-LPN assumption, which we reduce to standard LPN.
The reduction incurs a quadratic loss in noise rate, which is significantly more aggressive than for the lattice analogue~\cite{C:KLSS23b}.
A matching attack shows that this loss is inherent to the scheme rather than an artifact of the proof when all but two users collude.
\item \emph{Parameters.}
The $\KAHE$ balances LPN dimension and code rate.
Asymptotically, we characterize the optimum in closed form: one third of the communication overhead is due to the key sharing, two thirds due to the $\KAHE$ rate loss, decaying roughly as $\kC^{-1/3}$ in the input dimension $\kC$.
\item \emph{Field size and CRT.} Aggregation without overflow forces a large field size $\q$, with which the overhead of the $\KAHE$ grows.
We decouple the $\KAHE$ modulus from the aggregation modulus via a Chinese remainder theorem (CRT) decomposition, improving communication efficiency.
\item \emph{Comparison with information-theoretic masking.} For moderate to high collusion, the protocol requires less communication than direct secret sharing of the inputs as in SwiftAgg+~\cite{jahani2022swiftagg+}.
The break-even point reduces as $\kC$ grows.
\end{enumerate}

Python code for parameter optimization is available~\cite{code}.

\section{Preliminaries}\label{sec:prelim}
\textbf{Notation.} Sets are denoted by calligraphic letters. Vectors and matrices are denoted by bold lowercase and uppercase letters. For a positive integer $n$, we write $[n] \define \{1, \ldots, n\}$.
For prime $\q$, let $\F_\q$ denote the finite field of cardinality $\q$ and $\F_{\q^\extdeg}$ the extension field of degree $\extdeg\in\NN$.

\textbf{Probabilities.}
We write $x \getrand \mathcal{D}$ for sampling from a distribution $\mathcal{D}$, and $x \getrand \mathcal{S}$ for uniform sampling from a set $\mathcal{S}$.
$\Bern{\p}$ denotes the $\q$-ary Bernoulli distribution over $\F_\q$ with noise rate $\p$.
We use the $\q$-ary variant of the piling-up lemma below, which follows as in the binary case~\cite{C:EssKubMay17}.
\begin{lemma}\label{lem:piling}
Let $\eS_i \sim \Bern{\p}$ be independent for all $i\in[\Nuser]$.
Then, $\sum_{i\in[\Nuser]} \eS_i \in \F_\q$ follows $\Bern{\ppile{\Nuser}}$, where
\[
\ppile{\Nuser} \define \tfrac{\q-1}{\q}\left(1 - \big(1-\tfrac{\q}{\q-1}\,\p\big)^\Nuser\right).
\]
\end{lemma}
For independent $\eS \sim \Bern{\p}$ and $\eS'\sim\Bern{\p'}$, we define
\begin{align}
\postpp{\p}{\p'} &\define \Prob{\eS \neq 0 \mid \eS+\eS' = 0} \nonumber\\
&= \frac{\p\, \p'}{(\q-1)(1-\p)(1-\p') + \p\,\p'}; \label{eq:postrate}
\end{align}
conditioned on $e+e'=0$, $e$ follows $\Bern{\postpp{\p}{\p'}}$.
We write $\postp{\p} \define \postpp{\p}{\p}$ and note that $\postp{\p} \approx \p^{2}/(\q-1)$ for small $\p$.

\textbf{Coding theory.}
Let $\C$ be an $[\nC,\kC]_\q$ code, i.e., a $\q$-ary linear code of length $\nC$ and dimension $\kC$, with rate $R = \kC/\nC$.
We associate $\C$ with an encoder $\C.\enc\colon \F_\q^{\kC} \to \C$ and a decoder $\C.\dec\colon \F_\q^{\nC} \to \F_\q^{\kC}$.
The $\q$-ary symmetric channel with crossover probability $P$ adds an error distributed as $\Bern{P}$; its capacity is $1-\qentropy{P}$, where $\qentropy{\cdot}$ denotes the $\q$-ary entropy function
\[
\qentropy{P} \define P\,\log_\q(\q-1) - P\log_\q P - (1-P)\log_\q (1-P).
\]

\textbf{Learning parity with noise.} From the perspective of coding theory, the $\q$-ary LPN assumption asserts that a noisy codeword of a random $[\n,\m]_\q$ code is indistinguishable from random.
\begin{definition}[LPN]
Let $\m$ be the secret dimension, $\n$ the number of samples, and $\plpn$ the noise rate.
Sample $\sV\getrand\F_\q^\m$, $\eV\getrand\Bern{\plpn}^\n$, $\uV\getrand\F_\q^\n$ and $\aM\getrand\DistA$, where $\DistA$ is a distribution over $\F_\q^{\m\times\n}$, taken to be uniform unless stated otherwise.
The $\LPNattack{\plpn}{\m,\n}$ assumption asserts that $(\aM,\,\sV\aM+\eV)$ is computationally indistinguishable from $(\aM,\,\uV)$, written
\[
(\aM,\,\sV\aM+\eV)
\capprox
(\aM,\,\uV).
\]
\end{definition}
For the small noise rates in this work, information set decoding (ISD) has cost at least $2^{-\m\log_2(1-\plpn)}$~\cite{C:EssKubMay17}, where we conservatively ignore polynomial factors.
Using $\ln(x)\leq x-1$, we get that $\m = \lceil\lambda \ln2/\plpn\rceil$ is sufficient for $\lambda$ bits of security, i.e., the best known attacks require at least $2^{\lambda}$ operations.

\section{System Model and Aggregation Protocol}\label{sec:framework}
\noindent\textbf{Setting.}
A single server aggregates the inputs of $\Nuser$ users, indexed by $\uidx \in [\Nuser]$, using a committee of $\Ncm$ parties formed as a subset of the users.
User $\uidx$ holds a $\kinfo$-dimensional private input $\uinput\in\{0,\ldots,\qints-1\}^\kinfo$, which we embed into $\F_\q^{\kinfo}$ with $\q > \Nuser(\qints-1)$, so that $\uinputagg \define \sum_{\uidx\in[\Nuser]}\uinput$ is computed over the integers without overflow.

\noindent\textbf{Threat model.}
The server is honest-but-curious and colludes with a set $\colluser$ of at most $\coll \leq \Nuser - 2$ users.
Colluding parties follow the protocol but pool their views.
The protocol is \emph{secure under collusion} if the joint view of the server and $\colluser$ reveals nothing about $\{\uinput\}_{\uidx\notin\colluser}$ beyond $\uinputagg$ and the inputs of $\colluser$.
Unlike~\cite{C:BGLRS25}, we do not consider malicious deviations; the corresponding countermeasures are orthogonal to this work.

\noindent\textbf{KAHE.}
Following \cite{C:KarPol25,C:BGLRS25}, we rely on key- and message-additive homomorphic encryption (KAHE), which consists of
\vspace{-0.06cm}
\begin{itemize}
\item $\SetupKAHE() \rightarrow \Am$, public value derived from a seed.
\item $\KGenKAHE() \rightarrow \ukey\in\F_\q^\m$, secret key for each $\uidx$.
\item $\EncKAHE(\Am, \uinput, \ukey) \rightarrow\ctcrt\in\F_\q^\n$,  encryption of $\uinput\in\F_\q^\kinfo$.
\item $\AddKAHE(\{\ctcrt\}_{\uidx \in [\Nuser]}) \rightarrow\ctagg\in\F_\q^\n$, aggregated ciphertext.
\item $\DecKAHE(\Am,\ctagg, \ukeyagg) \rightarrow \uinputagg\in\F_\q^\kinfo$.
\end{itemize}
$\KAHE$ is key- and message-additive:  For $\ctagg = \AddKAHE(\{\ctcrt\}_{\uidx\in[\Nuser]})$ and $\ukeyagg \define \sum_{\uidx\in[\Nuser]} \ukey$, $\DecKAHE$ yields $ \uinputagg = \sum_{\uidx\in[\Nuser]} \uinput$.
As the server learns $\ukeyagg$, $\KAHE$ must reveal nothing about individual inputs beyond their sum, even given $\ukeyagg$ (cf.~\Cref{def:semantic_secure}).

\noindent\textbf{Key sharing.}
As in \cite{C:KarPol25}, keys are aggregated via McEliece--Sarwate secret sharing~\cite{mceliece1981sharing}.
To decouple the committee size $\Ncm$ from the modulus $\q$, we perform secret sharing over the extension field $\F_{\q^\extdeg}$, where $\extdeg = \lceil \log_\q \Ncm \rceil$.
To identify $\F_\q^{\extdeg}$ with $\F_{\q^\extdeg}$, any $\F_\q$-linear isomorphism $\varphi$ can be used, for instance
\[
\varphi\colon \F_\q^\extdeg \to \F_{\q^\extdeg}, \quad (x_1,\ldots,x_\extdeg) \mapsto \sum_{\gamma\in[\extdeg]} x_\gamma \beta_\gamma,
\]
where $\beta_1,\ldots,\beta_\extdeg$ denote an $\F_\q$-basis of $\F_{\q^\extdeg}$.
Applying $\varphi$ blockwise and zero-padding if necessary, user $\uidx$ splits $\ukey\in\F_\q^\m$ into $\Ncm-\coll$ parts of length $\mshare \define \lceil\m/(\extdeg(\Ncm-\coll))\rceil$ over $\F_{\q^\extdeg}$, i.e., $\ukeypart\in\F_{\q^\extdeg}^{\mshare}$, samples
$\sskeypart \getrand \F_{\q^\extdeg}^{\mshare}$ for $\pidx\in[\coll]$, and forms
\[
\poly(x) = \sum_{\pidx=1}^{\coll}\sskeypart\, x^{\pidx-1} + \sum_{\pidx=1}^{\Ncm-\coll} \ukeypart\, x^{\coll+\pidx-1} \in \F^{\mshare}_{\q^\extdeg}[x].
\]
Committee member $\cmidx\in[\Ncm]$ holds a distinct evaluation point $\eval \in \F_{\q^\extdeg}$ and receives key shares $\share \in \F_{\q^\extdeg}^{\mshare}$ from all users.
The $\Ncm$ committee members aggregate the users' shares as $\aggshare\define \sum_{\uidx\in[\Nuser]} \share$, from which the server obtains the polynomial $\aggpoly(x) = \sum_{\uidx\in[\Nuser]} \poly(x)$ of degree at most $\Ncm-1$ via interpolation;
its $\Ncm-\coll$ highest-degree coefficients are $\sum_{\uidx\in[\Nuser]} \ukeypart$, so $\ukeyagg$ follows by applying $\varphi^{-1}$.
No further information on the individual $\ukey$ is leaked.
As the evaluation points are distinct, any $\coll$ key shares are uniform and independent of $\ukey$.
Since the randomness is placed in the low-degree coefficients for $\poly$, this holds also for $\eval=0$.
Each key share consists of $\m/(\Ncm-\coll)$ symbols of $\F_\q$, independently of $\extdeg$, so the extension degree does not affect the communication cost.

Instead of communicating with committee members directly, the users may send their encrypted shares to the server, which relays them.
To this end, we use a $\PKE$ in the standard hybrid KEM/DEM mode: the share is encrypted under a symmetric key that is encapsulated for member $\cmidx$ via their public key $\cmpkey$.
The corresponding secret key $\cmskey$ allows member $\cmidx$ to first decapsulate the symmetric key, then decrypt the share.
Using, e.g., a code-based $\PKE$, the overhead due to the KEM is negligible compared to the key-sharing cost itself.

\begin{figure}[t]
    \centering
    \resizebox{\linewidth}{!}{% Preamble:
% \usepackage{tikz}
% \usetikzlibrary{arrows.meta,positioning,calc}

\begin{tikzpicture}[
  font=\large,
  >=Latex,
  node distance=10mm and 22mm,
  outer/.style={draw, rounded corners=3mm, thick, inner sep=7pt, align=center},
  proc/.style={draw, rounded corners=2mm, thick, fill=white, inner sep=5pt, align=left, minimum width=56mm},
  client/.style={outer, fill=blue!6, minimum width=63mm, minimum height=57mm},
  server/.style={outer, fill=gray!12, minimum width=63mm, minimum height=57mm},
  dec/.style={outer, fill=green!8, minimum width=63mm, minimum height=57mm},
  arr/.style={->, thick},
  lab/.style={pos=0.5,fill=white, inner sep=1.2pt},
  labX/.style={pos=0.9,fill=white, inner sep=1.2pt}
]

% ------------------------------------------------------------------
% CLIENT (single representative client)
\node[client] (C) {};
\node[font=\bfseries] at ($(C.north)+(0,-5mm)$) {\Large User $\uidx$};

\node[proc, anchor=south] (C1) at ($(C.north)+(0,-30mm)$)
{\textbf{\ref{step:keys} Key sharing}\\[.8mm]
$\ukey \leftarrow \KGenKAHE()$\\
$\ctcm \leftarrow \EncPKE(\share,\cmpkey)$};

\node[proc, anchor=north] (C2) at ($(C1.south)+(0,-7mm)$)
{\textbf{\ref{step:enc} Encryption}\\[.8mm]
$\ctcrt \leftarrow \EncKAHE(\Am,\uinput,\ukey)$};

% ------------------------------------------------------------------
% SERVER
\node[server, right=9mm of C] (S) {};
\node[font=\bfseries] at ($(S.north)+(0,-5mm)$) {\Large Server};

\node[proc, anchor=south, minimum width=54mm] (S1) at ($(S.north)+(0,-30mm)$)
{\textbf{\ref{step:keyrec} Key reconstruction}\\[.8mm]
$\ukeyagg \gets$ interpolate $\{\aggshare\}_{\cmidx\in[\Ncm]}$};

\node[proc, anchor=north, minimum width=54mm] (S2) at ($(S1.south)+(0,-7mm)$)
{\textbf{\ref{step:inputrec} Aggregation}\\[.8mm]
$\ctagg \leftarrow \AddKAHE(\{\ctcrt\}_{\uidx\in[\Nuser]})$\\
$\uinputagg \leftarrow \DecKAHE(\Am,\ctagg,\ukeyagg)$};

% ------------------------------------------------------------------
% DECRYPTOR
\node[dec, right=14mm of S] (D) {};
\node[font=\bfseries] at ($(D.north)+(0,-5mm)$) {\Large Committee member $\cmidx$};

\node[proc, anchor=south] (D1) at ($(D.north)+(0,-30mm)$)
{\textbf{\ref{step:keyagg} Share aggregation}\\[.8mm]
$\share \leftarrow \DecPKE(\ctcm,\cmskey)$\\
$\aggshare\leftarrow \sum_{\uidx\in[\Nuser]} \share$}; 

% ------------------------------------------------------------------
% COMMUNICATION ARROWS
\draw[arr] ([yshift=6mm]C1.east) -- node[labX] {$\ctcm$} node[pos=0.475, above=-.5mm] {relay} ([yshift=6mm]D1.west);
% \draw[arr] ([yshift=-4mm]D1.west) -- node[lab] {$\aggshare$} (S1.east);
\draw[arr] (C2.east) -- node[lab] {$\ctcrt$} (S2.west |- C2.east);
% \draw[arr] (S1.east) -- node[lab] {$\ctd$} (D1.west);
% \draw[arr] (D2.west) -- node[lab] {$\ukeyagg$} (S2.east);

\coordinate (aggstart) at ([yshift=-4mm]D1.west);

\draw[arr]
  (aggstart)
  -- node[lab] {$\aggshare$}
  (S1.east |- aggstart);

% ------------------------------------------------------------------
% Optional: short "one-time" badge near the decryptor arrow (kept minimal)
% \node[draw, rounded corners=1.8mm, thick, fill=yellow!10, inner sep=3pt, font=\footnotesize, above=4mm of $(D1.west)!0.5!(S1.east)$] (badge)
% {init only (one-time)};

\end{tikzpicture}}
    \caption{Protocol of \Cref{sec:framework}.
    After key sharing and reconstruction (Steps~1\,--\,3) are run once, an arbitrary number of rounds of encryption and aggregation (Steps~4\,--\,5) can be performed. \vspace{-.2cm}}
    \label{fig:schematic}
\end{figure}

\noindent\textbf{The aggregation protocol}  proceeds as shown in \Cref{fig:schematic}.
\begin{enumerate}[label={\arabic*)}, leftmargin=*, itemsep=0.3ex, topsep=2pt]
\item \label{step:keys} \emph{Key sharing:} User $\uidx$ samples $\ukey \gets \KGenKAHE()$ and sends $\ctcm = \EncPKE(\share, \cmpkey)$ to the server, which relays it to committee member $\cmidx\in[\Ncm]$.
\item \label{step:keyagg} \emph{Share aggregation:} Member $\cmidx$ decrypts $\ctcm$ and returns the aggregated shares $\aggshare = \sum_{\uidx\in[\Nuser]} \share$ to the server.
\item \label{step:keyrec} \emph{Key reconstruction:}
The server obtains $\aggpoly(x)$ by interpolating $\{\aggshare\}_{\cmidx\in[\Ncm]}$ and reads off $\ukeyagg$.
\item \label{step:enc} \emph{Encryption:} User $\uidx$ sends $\ct_\uidx = \EncKAHE(\Am,\uinput,\ukey)$.
\item \label{step:inputrec} \emph{Aggregation:} The server aggregates the ciphertexts $\ctagg = \AddKAHE(\{\ct_\uidx\}_{\uidx\in[\Nuser]})$ and retrieves the aggregate of the users' inputs $\uinputagg = \DecKAHE(\Am,\ctagg,\ukeyagg)$.
\end{enumerate}

\begin{remark}[Multiple rounds of aggregation]\label{rem:one_time}
Steps \ref{step:enc} and \ref{step:inputrec} can be repeated for an arbitrary number of inputs, e.g., across FL training rounds.
Each round requires a fresh $\Am$ but may reuse the same keys, such that the committee is involved only once (steps 1\,--\,3).
This amortizes the cost of providing the server with $\ukeyagg$, but inherently requires a fixed set of users, as, e.g., in cross-silo FL~\cite{kairouz2021advances}.
To achieve dynamic participation, OPA~\cite{C:KarPol25} involves a fresh committee in each round, while Willow~\cite{C:BGLRS25} replaces secret-sharing with threshold additively homomorphic encryption; a code-based instantiation of the latter is an open problem.
\end{remark}

The security of the protocol follows from the security of its components, analogously to \cite[Thm.~3]{C:KarPol25} and \cite[Thm.~2]{C:BGLRS25}.

\begin{proposition}\label{prop:protocol}
Let $\coll < \Ncm$, let $\PKE$ be semantically secure, and let $\KAHE$ be semantically secure for $\Nuser-\coll$ users under leakage of $\ukeyagg$ (Def.~\ref{def:semantic_secure}).
Then, the protocol is \emph{secure under collusion.}
\end{proposition}
\begin{proof}[Proof sketch]
For $\Nuser-\coll \leq 1$, the claim is immediate, since the remaining input is implied by the aggregate and the colluders' inputs; assume at least two non-colluding users.
By semantic security of $\PKE$, ciphertexts addressed to non-colluding committee members do not add to the colluders' view, which therefore contains at most $\coll$ key shares for any non-colluding user, together with the coefficients of $\aggpoly(x)$.
By design of the secret-sharing scheme, these shares are jointly uniform subject to consistency with $\aggpoly(x)$, so their distribution depends on the non-colluding keys only through $\ukeyagg$.
The remaining view consists of the KAHE ciphertexts and $\ukeyagg$, which is covered by \Cref{def:semantic_secure} after subtracting the colluders' contributions.
\end{proof}

The following section introduces an LPN-based $\KAHE$ that satisfies the requirement of \Cref{prop:protocol}; its communication cost and performance are analyzed in \Cref{sec:comm_cost}.

\section{Key-Additive Encryption from LPN}\label{sec:lpn}
Following Regev's blueprint~\cite{STOC:Regev05} in the Hamming metric, as \cite{C:BGLRS25} does for Ring-LWE, yields the following $\KAHE$.

\begin{definition}[LPN-based $\KAHE$]\label{def:LPN_KAHE}
For an $[\nC,\kC]_\q$ code $\C$, set%
\begin{itemize}
    \item $\SetupKAHE() \rightarrow \aM$ sampled from $\DistA$,
    \item $\KGenKAHE() \rightarrow \sV$ sampled uniformly at random from $\F_\q^\m$,
    \item $\EncKAHE(\aM, \mV, \sV) \rightarrow \sV\aM + \eV + \C.\enc(\mV)$ with $\eV\getrand\Bern{\pKAHE}^\n$,
    \item $\AddKAHE( \{\ctVu\}_{\uidx\in[\Nuser]}) \rightarrow  \sum_{\uidx\in[\Nuser]} \ctVu$, and
    \item $\DecKAHE(\aM, \ctaggV, \saggV) \rightarrow \C.\dec(\ctaggV - \saggV\aM)$.
\end{itemize}
\end{definition}
\begin{proposition}\label{prop:rate}
Let $\varepsilon> 0$, and $\beta < 1/2$.
Set $\saggV \define \sum_{\uidx\in[\Nuser]} \sVu$ and $\ctaggV \define \AddKAHE( \{\ctVu\}_{\uidx\in[\Nuser]})$.
Then, for sufficiently large $\nC$ and rate
$\Rate = \kinfo/\n = 1- \qentropy{\ppile{\Nuser}} - \varepsilon$, there exists a code $\C$ such that $\C.\enc$ and $\C.\dec$ require $\bigO{\nC\log\nC}$ operations and $\DecKAHE(\aM, \ctaggV, \saggV) = \sum_{\uidx\in[\Nuser]} \mVu$ with probability $1-\bigO{2^{-\n^\beta}}$.
\end{proposition}
\begin{proof}
$\ctaggV-\saggV\aM = \C.\enc(\maggV) + \eaggV$ with $\eaggV  \sim \Bern{\ppile{\Nuser}}^\n$ by \Cref{lem:piling}, i.e., the output of a $\q$-ary symmetric channel with crossover probability $\ppile{\Nuser}$.
Successive cancellation decoding of polar codes in time $\bigO{\n\log\n}$ is capacity achieving~\cite{csacsouglu2009polarization}.
\end{proof}

Keys and ciphertexts are generated individually, so colluding users can only subtract their contributions.
As in \cite{C:BGLRS25}, semantic security against $\coll$ colluders thus reduces to security for the $\Nuser-\coll$ honest users under leakage of their key sum.

\begin{definition}\label{def:semantic_secure}
Let $\aM\getrand \SetupKAHE()$, $\sVu\getrand\KGenKAHE()$, and $\saggV=\sum_{\uidx\in[\Nuser']} \sVu$.
$\KAHE$ is \emph{semantically secure for $\Nuser'$ users under leakage of the key sum}, if for all $\mVu[1],\ldots,\mVu[\Nuser'],\mVu[1]',\ldots,\mVu[\Nuser']'$ with $\sum_{\uidx\in[\Nuser']}\mVu=\sum_{\uidx\in[\Nuser']}\mVu'$,
\[
\big(\aM, \saggV, \ctVu[1],\ldots, \ctVu[\Nuser']\big)
\capprox
\big(\aM, \saggV, \ctVu[1]',\ldots, \ctVu[\Nuser']'\big),
\]
where
$\ctVu\getrand \EncKAHE(\aM,\mVu,\sVu)$ and $\ctVu'\getrand \EncKAHE(\aM,\mVu',\sVu)$.
\end{definition}

Leakage of the key sum $\saggV$ reveals also $\eaggV = \sum_{\uidx\in[\Nuser]} \eVu$, i.e., a noisy version of each $\KAHE$ error.
To prove \Cref{def:semantic_secure}, it is therefore natural to take an intermediate step via the \emph{Hint}-LPN assumption, which reveals a noisy version of the LPN error.
Willow~\cite{C:BGLRS25} uses the lattice analogue, Hint-LWE~\cite{C:KLSS23b}.
\begin{definition}
Let $\eV\getrand\Bern{\pe}^\nC, \fV\getrand\Bern{\pf}^\nC$, $\sV\getrand\F_\q^{\m}$, $\aM\getrand\DistA$, $\uV\getrand\F_\q^\n$.
The $\HintLPN$ assumption asserts that
\[
(\aM,\,\sV\aM+\eV,\,\eV+\fV) \capprox (\aM,\, \uV,\,\eV+\fV)\,.
\]
\end{definition}

\noindent The role of $\eV$ and $\fV$ in $\HintLPN$ can be switched:
Given a $\HintLPN$ instance $(\aM,\,\bV,\,\hV)$,
$(\aM,\,\hV-\bV,\,\hV)$ is an instance of $\HintLPN[\pf\!,\, \pe]$; this symmetry motivates $\pe = \pf$.
The key step in basing $\HintLPN[\p,\,\p]$ on $\LPN$ is the decomposition below.
\begin{lemma}\label{lem:cover}
Let $\eS,\fS\getrand\Bern{\p}$ with $\p \leq \frac{q-1}{q}$ and let $\plpn \leq \postp{\p}$.
Then, there exist distributions $\mathcal{D}_h$ such that $(\eS'+\tS,\,\hS)$ with $\eS' \getrand \Bern{\plpn}$, $\hS \getrand \Bern{\ppile{2}}$, $t \getrand \Dset_{h}$
follows the same distributions as $(\eS,\,\eS + \fS)$.
\end{lemma}
\begin{proof}
By \Cref{lem:piling}, $\hS$ and $\eS+\fS$ are identically distributed, so it remains to match the conditional distributions.
Bayes' law and \Cref{lem:piling} determine
$P(\xS\given \hS) \define \Prob{\eS = \xS \given \eS+\fS = \hS}$ as
\[
P(\xS\given \hS) =
\begin{cases}
1-\postp{\p}&\text{for  } \hS = 0,\,\xS = 0, \\
\tfrac{\postp{\p}}{\q-1}&\text{for  } \hS = 0,\,\xS \neq 0,\\
\tfrac{(\q-1)(1-\p)}{2(\q-1) - \q\p}&\text{for  } \hS \neq 0,\,\xS \in\{0,\hS\},\\
\tfrac{\p}{2(\q-1) - \q\p}&\text{for  } \hS \neq 0,\,\xS \notin\{0,\hS\},
\end{cases}
\]
which is minimal for $\xS \neq 0$ and $\hS = 0$.
For all $\xS\in\F_\q$ and $\eS'\getrand\Bern{\plpn}$ independent of $\tS\getrand\Dset_{\hS}$, we have
\[
\Prob{\eS'+\tS = \xS} = \Big(1-\tfrac{\q}{\q-1}\plpn\Big)\Prob{\tS=\xS} + \tfrac{\plpn}{\q-1}.
\]
For $\plpn < \tfrac{\q-1}{\q}$ this map is invertible, so $\Dset_h$ is determined by
$\Prob{\tS=\xS} = \bigl((\q-1)\,P(\xS\given \hS) - \plpn\bigr)/(\q - 1 - \q\plpn)$.
Non-negativity needs $\min_{\xS,\hS}\! P(\xS\given \hS) \geq \plpn/(\q-1)$, i.e., $\plpn \leq \postp{\p}$.
\end{proof}

\begin{theorem}\label{thm:LPN_to_HintLPN}
Let $\p \leq \frac{q-1}{q}$ and $\plpn = \postp{\p}$.
Then there exists an efficient reduction from $\LPN$ to $\HintLPN[\p,\,\p]$.
\end{theorem}
\begin{proof}
Let an $\LPN$ instance $(\aM,\bV)$ be given.
For $\vidx\in[\n]$, sample $\hS_\vidx \getrand \Bern{\ppile{2}}$ and $\tS_\vidx \getrand \mathcal{D}_{h_i}$ as in \Cref{lem:cover} and output $(\aM,\bV+\tV,\hV)$.
For $\bV=\sV\aM+\eV'$, \Cref{lem:cover} gives that $(\aM,\bV+\tV,\hV)$ is distributed as $(\aM,\sV\aM+\eV,\eV+\fV)$ with $\eV,\fV\getrand\Bern{\p}^\n$;
if $\bV$ is uniform, so is $\bV+\tV$.
\end{proof}
By \Cref{thm:LPN_to_HintLPN}, the $\HintLPN[\p,\,\p]$ assumption follows from the $\LPN$ assumption, albeit with reduced noise rate $\plpn \approx \p^2/(\q-1)$.
Along the lines of the hybrid argument in \cite{C:BGLRS25}, this yields the security of $\KAHE$ under LPN.

\begin{theorem}\label{thm:KAHE_reduction}
Let $\p \leq \frac{\q-1}{\q}$. Under the $\LPN[\postp{\p}]$ assumption, the KAHE of \Cref{def:LPN_KAHE} satisfies \Cref{def:semantic_secure} for every $N'\geq 2$.
In particular, the protocol of \Cref{sec:framework} is secure against collusion of up to $\min\{\Nuser-2,\Ncm-1\}$ users with the server.
\end{theorem}
\begin{proof}[Proof sketch]
Hybrid distributions $\Hyb{1},\ldots\Hyb{\Nuser'}$ successively replace $\ctVu[\hidx]$ with uniformly random vectors, which is compensated in $\ctVu[1]$ to keep consistency with $\maggV$.
$\Hyb{\Nuser'}$ depends on the messages only via $\maggV$ and is reached by both message vectors of \Cref{def:semantic_secure}.
We construct a distinguisher for $\HintLPN[\p,\,\p]$ from a distinguisher between $\Hyb{\hidx-1}$ and $\Hyb{\hidx}$, and the claim follows from \Cref{thm:LPN_to_HintLPN}.
The full proof is in \extversion.
\end{proof}

The attack below shows that \Cref{thm:KAHE_reduction} is tight for $\coll = \Nuser-2$ up to a constant factor in the number of LPN samples.

\begin{proposition}\label{prop:KAHE_solver}
Let $\varepsilon>0$, $\Nuser' = \Nuser - \coll$ and $\plpn'=\postpp{\p}{\ppile{\Nuser'-1}}$.
A server that colludes with $\coll$ users and is able to break the search-$\LPNattack{\plpn'}{\m,\n'}$ assumption with $\n' = \lfloor(1-\ppile{\Nuser'}-\varepsilon)\n\rfloor$ samples breaks \Cref{def:semantic_secure} except with negligible probability.
\end{proposition}
\begin{proof}[Proof sketch]
Subtracting the contributions of the $\coll$ colluders, the server computes the honest users' error sum
\[
\eaggV = \sum\nolimits_{\uidx\in[\Nuser']} \ctVu - \saggV\aM - \C.\enc(\maggV) = \sum\nolimits_{\uidx\in[\Nuser']} \eVu
\]
and sets $\Iset = \{\vidx : \eaggS_\vidx = 0\}$, where $\size{\Iset}\geq\n'$ except with probability $\exp(-2\varepsilon^{2}\n)$ by Hoeffding's inequality.
The server forms $(\aM_\Iset,\bV_\Iset)$ and $(\aM_\Iset,\bV'_\Iset)$ with $\bV = \ctVu[1] - \C.\enc(\mVu[1])$ and $\bV' = \ctVu[1] - \C.\enc(\mVu[1]')$.
By \Cref{eq:postrate}, the error $\eS_{1,\vidx}$ follows $\Bern{\plpn'}$ conditioned on $\eaggS_\vidx = 0$, so $(\aM_\Iset,\bV_\Iset)$ or $(\aM_\Iset,\bV'_\Iset)$ is an $\LPNattack{\plpn'}{\m,\n'}$ instance, depending on which message was encrypted.
The server runs the solver on both and decides correctly between $\mVu[1]$ and $\mVu[1]'$ based on which residual has smaller weight.
The full proof is deferred to \extversion.
\end{proof}

\begin{example}
Let $\Nuser = 10$, $\q=11$ and $\p=0.01$, so that $R < 1-\qentropy{\ppile{\Nuser}}=0.78$ is admissible.
\Cref{thm:KAHE_reduction} requires $\m = 8.7\cdot10^6$ to achieve $\lambda = 128$ bits of security.
This agrees with \Cref{prop:KAHE_solver} for $\coll=8$; for $\coll=2$, $\m = 1.2\cdot10^6$ would be sufficient for the attack to require at least $2^\lambda$ operations.
\end{example}

As positions with $\eaggS_i\neq0$ are simply discarded, the attack of \Cref{prop:KAHE_solver} may not be optimal.
Therefore, parameters are chosen according to \Cref{thm:KAHE_reduction} in the following.

\section{Communication Cost and Complexity}\label{sec:comm_cost}
\begin{figure}[t]
\centering
% C(p)/(k log2 q) vs. p for several input lengths; N = 10, z = 5, q = 11 (binary inputs).
\begin{tikzpicture}

\begin{axis}[%
scale only axis,
axis x line*=bottom,
xmode = log,
ymin = 1,
ymax = 5,
xmin=4e-4, 
xmax=1e-1, 
xlabel={\small $\KAHE$ noise rate $\p$},
ylabel={\small $C/(\kC\,\log_2\q)$},
ymajorgrids,
xmajorgrids,
grid style=dashed,
legend style={at={(1.02,0.5)},anchor=west,font=\footnotesize,legend cell align=left, align=left, draw=black},
ylabel near ticks,
ylabel style={yshift=-3pt},
xlabel near ticks,
width=6.cm, 
height = 3.4cm,
xticklabel style = {font=\scriptsize},
yticklabel style = {font=\scriptsize},
legend columns=1,
% scaled ticks=false,
% legend to name={opt_legend},
legend image post style={xscale=0.5},
]

\addplot [norm10] % 
  table[row sep=crcr]{%%
1.360e-02 10.722\\ % m=4.7e+06
1.458e-02 9.528\\ % m=4.1e+06
1.562e-02 8.495\\ % m=3.5e+06
1.675e-02 7.604\\ % m=3.1e+06
1.795e-02 6.837\\ % m=2.7e+06
1.924e-02 6.179\\ % m=2.3e+06
2.062e-02 5.616\\ % m=2.0e+06
2.210e-02 5.137\\ % m=1.7e+06
2.368e-02 4.732\\ % m=1.5e+06
2.538e-02 4.393\\ % m=1.3e+06
2.720e-02 4.113\\ % m=1.1e+06
2.916e-02 3.886\\ % m=9.8e+05
3.125e-02 3.706\\ % m=8.5e+05
3.349e-02 3.571\\ % m=7.4e+05
3.590e-02 3.477\\ % m=6.4e+05
3.847e-02 3.421\\ % m=5.5e+05
4.123e-02 3.403\\ % m=4.8e+05
4.419e-02 3.423\\ % m=4.2e+05
4.737e-02 3.480\\ % m=3.6e+05
5.077e-02 3.577\\ % m=3.1e+05
5.441e-02 3.716\\ % m=2.7e+05
5.831e-02 3.902\\ % m=2.3e+05
6.250e-02 4.140\\ % m=2.0e+05
6.699e-02 4.440\\ % m=1.7e+05
7.179e-02 4.811\\ % m=1.5e+05
7.695e-02 5.270\\ % m=1.3e+05
8.247e-02 5.836\\ % m=1.1e+05
8.839e-02 6.537\\ % m=9.4e+04
9.473e-02 7.409\\ % m=8.1e+04
1.015e-01 8.504\\ % m=7.0e+04
1.088e-01 9.893\\ % m=6.0e+04
1.166e-01 11.677\\ % m=5.1e+04
1.250e-01 14.000\\ % m=4.4e+04
1.340e-01 17.078\\ % m=3.7e+04
1.436e-01 21.230\\ % m=3.2e+04
1.539e-01 26.952\\ % m=2.7e+04
};
\addlegendentry{$\kC = 10^{6}$}
 
\addplot [norm10, densely dashed] % 
  table[row sep=crcr]{%%
% 1.611e-04 6838.256\\
% 1.983e-04 4511.574\\
% 2.441e-04 2976.601\\
% 3.006e-04 1963.948\\
% 3.700e-04 1295.889\\
% 4.556e-04 855.169\\
% 5.609e-04 564.432\\
% 6.905e-04 372.641\\
% 8.501e-04 246.127\\
% 1.047e-03 162.676\\
% 1.289e-03 107.634\\
1.586e-03 71.333\\
1.953e-03 47.397\\
2.405e-03 31.617\\
2.960e-03 21.218\\
3.645e-03 14.371\\
4.487e-03 9.867\\
5.524e-03 6.912\\
6.801e-03 4.981\\
8.373e-03 3.729\\
1.031e-02 2.931\\
1.269e-02 2.438\\
1.562e-02 2.156\\
1.924e-02 2.027\\
2.368e-02 2.018\\
2.916e-02 2.115\\
3.590e-02 2.325\\
4.419e-02 2.676\\
5.441e-02 3.234\\
6.699e-02 4.130\\
8.247e-02 5.638\\
1.015e-01 8.379\\
1.250e-01 13.922\\
% 1.539e-01 26.904\\
};
\addlegendentry{$\kC = 10^{7}$}

\addplot [norm10, dashdotted] % 
  table[row sep=crcr]{%%
3.221e-04 171.888\\ % m=8.5e+09
3.700e-04 130.501\\ % m=6.5e+09
4.251e-04 99.138\\ % m=4.9e+09
4.883e-04 75.371\\ % m=3.7e+09
5.609e-04 57.361\\ % m=2.8e+09
6.443e-04 43.714\\ % m=2.1e+09
7.401e-04 33.374\\ % m=1.6e+09
8.501e-04 25.539\\ % m=1.2e+09
9.766e-04 19.603\\ % m=9.3e+08
1.122e-03 15.107\\ % m=7.0e+08
1.289e-03 11.702\\ % m=5.3e+08
1.480e-03 9.123\\ % m=4.0e+08
1.700e-03 7.172\\ % m=3.1e+08
1.953e-03 5.695\\ % m=2.3e+08
2.244e-03 4.580\\ % m=1.8e+08
2.577e-03 3.738\\ % m=1.3e+08
2.960e-03 3.103\\ % m=1.0e+08
3.401e-03 2.627\\ % m=7.6e+07
3.906e-03 2.271\\ % m=5.8e+07
4.487e-03 2.006\\ % m=4.4e+07
5.154e-03 1.812\\ % m=3.3e+07
5.921e-03 1.673\\ % m=2.5e+07
6.801e-03 1.575\\ % m=1.9e+07
7.812e-03 1.511\\ % m=1.4e+07
8.974e-03 1.474\\ % m=1.1e+07
1.031e-02 1.459\\ % m=8.2e+06
1.184e-02 1.463\\ % m=6.2e+06
1.360e-02 1.484\\ % m=4.7e+06
1.562e-02 1.523\\ % m=3.5e+06
1.795e-02 1.578\\ % m=2.7e+06
2.062e-02 1.652\\ % m=2.0e+06
2.368e-02 1.746\\ % m=1.5e+06
2.720e-02 1.867\\ % m=1.1e+06
3.125e-02 2.018\\ % m=8.5e+05
3.590e-02 2.209\\ % m=6.4e+05
4.123e-02 2.454\\ % m=4.8e+05
4.737e-02 2.770\\ % m=3.6e+05
5.441e-02 3.185\\ % m=2.7e+05
};
\addlegendentry{$\kC = 10^{8}$}

\addplot [norm10, densely dotted] % 
  table[row sep=crcr]{%%
% 1.611e-04 69.379\\
% 1.983e-04 46.114\\
% 2.441e-04 30.765\\
3.006e-04 20.641\\
3.700e-04 13.963\\
4.556e-04 9.558\\
5.609e-04 6.654\\
6.905e-04 4.741\\
8.501e-04 3.480\\
1.047e-03 2.652\\
1.289e-03 2.109\\
1.586e-03 1.754\\
1.953e-03 1.525\\
2.405e-03 1.380\\
2.960e-03 1.292\\
3.645e-03 1.242\\
4.487e-03 1.220\\
5.524e-03 1.219\\
6.801e-03 1.234\\
8.373e-03 1.265\\
1.031e-02 1.311\\
1.269e-02 1.375\\
1.562e-02 1.459\\
1.924e-02 1.571\\
2.368e-02 1.719\\
2.916e-02 1.920\\
3.590e-02 2.198\\
4.419e-02 2.594\\
5.441e-02 3.181\\
6.699e-02 4.096\\
8.247e-02 5.617\\
1.015e-01 8.365\\
1.250e-01 13.913\\
1.539e-01 26.899\\
};
\addlegendentry{$\kC = 10^{9}$}

% \addplot [norm10, densely dotted] % 
%   table[row sep=crcr]{%%
% 3.221e-04 2.721\\ % m=8.5e+09
% 3.700e-04 2.309\\ % m=6.5e+09
% 4.251e-04 1.997\\ % m=4.9e+09
% 4.883e-04 1.761\\ % m=3.7e+09
% 5.609e-04 1.584\\ % m=2.8e+09
% 6.443e-04 1.450\\ % m=2.1e+09
% 7.401e-04 1.349\\ % m=1.6e+09
% 8.501e-04 1.274\\ % m=1.2e+09
% 9.766e-04 1.219\\ % m=9.3e+08
% 1.122e-03 1.178\\ % m=7.0e+08
% 1.289e-03 1.149\\ % m=5.3e+08
% 1.480e-03 1.129\\ % m=4.0e+08
% 1.700e-03 1.116\\ % m=3.1e+08
% 1.953e-03 1.108\\ % m=2.3e+08
% 2.244e-03 1.105\\ % m=1.8e+08
% 2.577e-03 1.106\\ % m=1.3e+08
% 2.960e-03 1.111\\ % m=1.0e+08
% 3.401e-03 1.118\\ % m=7.6e+07
% 3.906e-03 1.128\\ % m=5.8e+07
% 4.487e-03 1.142\\ % m=4.4e+07
% 5.154e-03 1.158\\ % m=3.3e+07
% 5.921e-03 1.177\\ % m=2.5e+07
% 6.801e-03 1.200\\ % m=1.9e+07
% 7.812e-03 1.228\\ % m=1.4e+07
% 8.974e-03 1.259\\ % m=1.1e+07
% 1.031e-02 1.297\\ % m=8.2e+06
% 1.184e-02 1.340\\ % m=6.2e+06
% 1.360e-02 1.392\\ % m=4.7e+06
% 1.562e-02 1.453\\ % m=3.5e+06
% 1.795e-02 1.525\\ % m=2.7e+06
% 2.062e-02 1.612\\ % m=2.0e+06
% 2.368e-02 1.717\\ % m=1.5e+06
% 2.720e-02 1.844\\ % m=1.1e+06
% 3.125e-02 2.001\\ % m=8.5e+05
% 3.590e-02 2.197\\ % m=6.4e+05
% 4.123e-02 2.444\\ % m=4.8e+05
% 4.737e-02 2.762\\ % m=3.6e+05
% 5.441e-02 3.180\\ % m=2.7e+05
% };
% \addlegendentry{$\kC = 10^{10}$}

\end{axis}

\begin{axis}[
scale only axis, 
width=6.cm, 
height=3.4cm,
xmode=log, 
xmin=4e-4, 
xmax=1e-1, 
ymin = 1,
ymax = 10,
axis x line*=top, 
axis y line=none,
xlabel={\small LPN dimension $\m$}, 
xlabel near ticks,
xtick={8.610e-02, 2.893e-02, 9.331e-03,2.970e-03,9.410e-04},
xticklabels={$10^{5}$, $10^{6}$, $10^{7}$, $10^{8}$, $10^{9}$},
xticklabel style={font=\scriptsize},
minor xtick = {0.08610173537608433, 0.06245173639494721, 0.051580964727186214, 0.044980476567672234, 0.04042329184904955, 0.0370314115112906, 0.03437868804987253, 0.032229756737255734, 0.030442528509152307, 0.028925436840010938, 0.020627911372665828, 0.016906554144959236, 0.01467471766490693, 0.013145818308772572, 0.012014187453557519, 0.01113289217663828, 0.010421356197264486, 0.009831205321869914, 0.009331405752589257, 0.006616376270434103, 0.0054088136069575425, 0.0046875659409011785, 0.0041947618758005345, 0.0038306759465153942, 0.003547524662021738, 0.0033191658177018414, 0.003129933590954228, 0.00296979265051354, 0.0021017884925777454, 0.0017167651859307752, 0.0014871042741980746, 0.00133031534065738, 0.001214546969159422, 0.0011245533068051023, 0.0010519996997023187, 0.0009918945027586382, 0.00094104164546756, 0.000665600377127678, 0.0005435268165475211, 0.00047074230850100806, 0.00042106564700744933, 0.0003843926898560409, 0.0003558886257595258, 0.0003329109791964817, 0.0003138774574969064, 0.00029777509772507913, 0.00021057715484758401, 0.00017194217202843418, 0.0001489097182387941, 0.00013319099650511155, 0.00012158759667757211, 0.00011256925829011454, 0.00010529966673552657, 9.927807387576615e-05},
xminorticks=true,
]
% \addplot [draw=none, forget plot] coordinates {(4e-5,1)};
\end{axis}

\end{tikzpicture}
%%%%%%%%%%%%%%%%%%%%%%%%%%%%%%%%
%Second Part of Latex include:
%end{figure}
\caption{Per-user communication as in \Cref{eq:Cp}, normalized by the payload $\kC\log_2\q$, for $\Nuser=\Ncm=10$, $\coll = 5$ and $\q=11$. \vspace{-.2cm}}
\label{fig:opt_k}
\end{figure}

\noindent\textbf{Communication for a single modulus.}
In the protocol of \Cref{sec:framework}, each user sends $\frac{\m\,\Ncm}{\Ncm-\coll}\log_2\q$ bits of key shares once and $\frac{\kC}{R}\log_2\q$ bits of $\KAHE$ ciphertexts.
Both terms are governed by a single free parameter, the noise rate $\p$: increasing $\p$ lowers the LPN dimension $\m$ but also the code rate $R$, so that the two terms trade off.
\Cref{fig:opt_k} illustrates the overall per-user communication cost, normalized by the payload $\kC\,\log_2\q$, as a function of $\p$.
As $\Nuser/(\Nuser-\coll) = 2$, direct secret sharing of the inputs as in \cite{jahani2022swiftagg+} would give a per-user communication $2\,\kC\log_2\q$; our protocol is cheaper for $\kC> 10^7$.
The following theorem characterizes the optimum without numerical search over $\p$.
\begin{theorem}\label{prop:opt}
Write $a \define \frac{\Ncm}{\Ncm-\coll}$, $\m(\p) \define \lambda\ln2/\postp{\p}$ and $R(\p) \define 1-\qentropy{\ppile{\Nuser}}$.
Fix $\Nuser,\Ncm,\q,\coll$ and $\lambda$, and let $\kC\to\infty$.
Then, the per-user communication cost
\begin{equation}\label{eq:Cp}
C(\p) \define \Big(a\,\m(\p) + \frac{\kC}{R(\p)}\Big)\log_2\q,
\end{equation}
has a unique minimizer
\begin{equation}\label{eq:popt}
\p^\star = \Big(\frac{6\ln2\,a\,\lambda\,(\q-1)\log\q}{\Nuser\,\kC\,\log \kC}\Big)^{1/3}(1+\smallO{1}),
\end{equation}
for which $\delta \define 1-R(\p^\star) = \tfrac{1}{3}\,\Nuser\,\p^\star\, \tfrac{\log\kC}{\log\q}\,(1+\smallO{1})$ and
\begin{equation}\label{eq:copt}
C(\p^\star) = \kC\log_2\q\,\Big(1+\tfrac{3}{2}\,\delta\,\bigl(1+\smallO{1}\bigr)\Big).
\end{equation}
\end{theorem}
\begin{proof}[Proof sketch]
We write $P \define \ppile{\Nuser}$ and minimize $\bar C(P) \define a\,\m(P) + \kC/R(P)$.
The key sharing cost is $a\,\m(P) = \frac{A}{P^2}(1+\bigO{P})$ with $A \define a\lambda\ln2\,(\q-1)\Nuser^2$ by \Cref{lem:piling} and \Cref{eq:postrate}.
Comparing with the trial point $P_0 = \kC^{-1/3}$ and using $\kC/R(P)\geq\kC(1+\qentropy{P})\geq\kC(1+P)$ confines the minimizer to $\Omega(\kC^{-1/3}\log^{-1/2}\kC)\leq P^\star\leq\bigO{\kC^{-1/3}\log\kC}$.
Hence $R(P^\star)\to1$ and $\qentropy[']{P^\star} = \log_\q\frac{(\q-1)(1-P^\star)}{P^\star} = \tfrac13\log_\q\kC\,(1+\smallO{1})$.
For the minimizer $P^\star$, stationarity requires
\[
\frac{2A}{(P^\star)^3} = \kC\,\qentropy[']{P^\star}\,(1+\smallO{1}),
\]
whose solution is unique and yields \Cref{eq:popt} via $\p^\star = P^\star/\Nuser\,(1+\smallO{1})$.
Since $\qentropy{P} = P\,\qentropy[']{P} - \log_\q(1-P)$, we have $\delta = P^\star\qentropy[']{P^\star}(1+\smallO{1})$, which gives the expression for $\delta$.
For the $\KAHE$, we get $\kC/R(P^\star) = \kC\bigl(1+\delta(1+\smallO{1})\bigr)$, and the balance condition shows that the key-sharing term equals $\tfrac{\delta}{2}\kC(1+\smallO{1})$, i.e., half of the $\KAHE$ overhead.
Together, this gives \Cref{eq:copt}. A full proof is given in \extversion.
\end{proof}

\Cref{prop:opt} shows that user communication is the raw payload $\kC\log_2\q$ plus a relative overhead of $\tfrac{3}{2}\delta$, two thirds of which is $\KAHE$ rate loss and one third the key sharing.
The collusion parameter enters this overhead only through $a^{1/3} = \big(\tfrac{\Ncm}{\Ncm-\coll}\big)^{1/3}$, whereas direct secret sharing pays $\tfrac{\coll}{\Nuser-\coll}$ on the entire input.
Moreover, the overhead vanishes as $\kC\to\infty$, matching \Cref{rem:one_time}.
While asymptotic, \Cref{prop:opt} yields an accurate prediction of the true cost for finite $\kC$.
\begin{example}\label{ex:copt}
Let $\Nuser = \Ncm = 10$, $\coll = 5$, $\lambda = 128$ and assume binary inputs of dimension $\kC = 10^{10}$.
Then, $\q = 11$ suffices and $a = 2$.
\Cref{prop:opt} predicts $\p^\star = 2.2\cdot10^{-3}$, $\m(\p^\star) = 1.8\cdot10^{8}$ and $\delta = 7.1\%$, i.e., a per-user communication of $1.11\,\kC\log_2\q$; numerical minimization of $C$ matches within $0.2\%$.
In contrast, quantization with $\qints = 2^{16}$ requires $\q = 655351$, which exceeds the regime of \Cref{prop:opt} and numerical minimization gives $C = 4.12\,\kC\log_2\q$; the scaling of the communication overhead with the modulus $\q$ is improved below.
\end{example}

\noindent\textbf{Aggregation over several small fields.}
By \Cref{prop:opt}, the communication overhead grows as $\big(\tfrac{\q-1}{\log^{2}\q}\big)^{1/3}$, which is minimized for $\q\approx 5$.
The use of an extension field allows choosing  $\q$ independent of the committee size $\Ncm$; however, aggregation without overflows nevertheless requires $\q > \Nuser(\qints-1)$.
We resolve this conflict via the Chinese remainder theorem (CRT).

Let $\crtmod[1],\ldots,\crtmod[\Ncrt]$ be distinct primes with $\q \define \prod_{\crtidx\in[\Ncrt]} \crtmod > \Nuser(\qints-1)$, and let $\uinputcrt\in\F_{\crtmod}^{\kC}$ be the reduction of $\uinput\in\Z_\q^{\kC}$ modulo $\crtmod$.
The protocol of \Cref{sec:framework} is run independently for each $\crtidx$, with its own LPN parameters $(\m_\crtidx,\p_\crtidx)$ and code rate $R_\crtidx$, yielding $\uinputcrtagg = \sum_{\uidx}\uinputcrt$; the server recovers $\uinputagg\in\Z_\q^{\kC}$ by the CRT.
Users maintain a single interaction with the server per input by sending the messages of all $\Ncrt$ instances jointly.
For each $\crtidx$, the full committee of size $\Ncm$ is used; $\extdegcrt = \lceil\log_{\crtmod} \Ncm\rceil$ allows assigning each member a distinct evaluation point $\evalcrt\in\F_{\crtmod^{\extdegcrt}}$.
The overall per-user communication is
\begin{equation}\label{eq:comm}
\sum\nolimits_{\crtidx\in [\Ncrt]} C_\crtidx \ \text{bit}
\quad\text{with}\quad
C_\crtidx = \Big(\frac{\m_\crtidx\,\Ncm}{\Ncm-\coll} + \frac{\kC}{R_\crtidx}\Big)\log_2\crtmod;
\end{equation}
the moduli, and the corresponding LPN parameters, are chosen to minimize this cost.
Security follows from that of each instance, and \Cref{prop:opt} applies to each summand of~\eqref{eq:comm}.

\begin{example}\label{ex:copt_crt}
Continuing \Cref{ex:copt} with $\qints = 2^{16}$, the CRT decomposition uses the moduli $\{2,3,5,7,11,13,23\}$, whose product exceeds $\Nuser(\qints-1)$.
Each instance uses $\p_\crtidx^\star$ between $7.7\cdot10^{-4}$ and $3.2\cdot10^{-3}$, $\m_\crtidx^\star \leq 1.9\cdot10^{8}$ and $\delta_\crtidx \leq 8\%$, so that~\eqref{eq:comm} evaluates to $1.11\,\kC\log_2\q$, matching the prediction via \Cref{prop:opt} within $0.3\%$.
The CRT decomposition thus achieves a relative overhead that is practically identical to the one for binary inputs in \Cref{ex:copt}, reducing communication cost by a factor of $3.7$ compared with a single modulus.
\end{example}

\noindent\textbf{Comparison.}
\Cref{fig:comparison} compares the communication cost per user of \Cref{eq:comm}, normalized by the payload of $\kC\log_2\q$ bits, with information-theoretically secure aggregation, such as SwiftAgg+~\cite{jahani2022swiftagg+}.
Direct secret sharing of the inputs requires $\frac{\Nuser}{\Nuser-\coll}\kC\log_2\q$ bits per user, whereas in our protocol the collusion parameter only affects the key sharing, which has input dimension $\m_\crtidx\ll\kC$.
At $\kC = 10^{10}$, our protocol requires less communication once $\coll/\Nuser$ is at least $0.10$, $0.22$ and $0.30$ for $\Nuser = 10$, $50$ and $100$, respectively.
These thresholds decrease as $\kC$ grows, since the overhead decays roughly as $\kC^{-1/3}$ by \Cref{prop:opt}, unlike for SwiftAgg+.
High collusion levels arise, e.g., in cross-silo FL with few participants~\cite{kairouz2021advances}, or when the server registers users under its control (sybil attacks~\cite{C:BGLRS25}).

\begin{figure}[t]
    \centering
    \begin{tikzpicture}%[baseline, trim axis left, trim axis right] % trim axis right, trim axis left,
\begin{axis}[%
scale only axis,
anchor=origin,
enlargelimits=false,
% ymode = log,
ymin = {.95},
ymax = {2.5},
xmin = 0,
xmax = 1,
xlabel={\small rel. collusion parameter $\coll/\Nuser$},
ylabel={\small $C/(\kC\,\log_2\q)$},
ymajorgrids,
xmajorgrids,
grid style=dashed,
legend style={at={(1.03,0.5)},anchor=west,font=\footnotesize,legend cell align=left, align=left, draw=black},
ylabel near ticks,
xlabel near ticks,
width = 4.7cm,
height = 3.1cm,
xticklabel style = {font=\scriptsize},
yticklabel style = {font=\scriptsize},
legend columns=1,
scaled ticks=false,
]

\addplot [normSwift]
  table[row sep=crcr]{%
0.0 1.000\\
0.01 1.010\\
0.02 1.020\\
0.03 1.031\\
0.04 1.042\\
0.05 1.053\\
0.09 1.099\\
0.13 1.149\\
0.17 1.205\\
0.21 1.266\\
0.25 1.333\\
0.29 1.408\\
0.33 1.493\\
0.37 1.587\\
0.41 1.695\\
0.45 1.818\\
0.49 1.961\\
0.53 2.128\\
0.57 2.326\\
0.61 2.564\\
0.65 2.857\\
0.69 3.226\\
0.73 3.704\\
0.77 4.348\\
0.8 5.000\\
0.81 5.263\\
0.82 5.556\\
0.83 5.882\\
0.84 6.250\\
0.85 6.667\\
0.86 7.143\\
0.87 7.692\\
0.88 8.333\\
0.89 9.091\\
0.9 10.000\\
0.91 11.111\\
0.92 12.500\\
0.93 14.286\\
0.94 16.667\\
0.95 20.000\\
0.96 25.000\\
0.97 33.333\\
0.98 50.000\\
};
\addlegendentry{secret sharing}

\addplot [norm10] 
  table[row sep=crcr]{%
0.0 1.090\\
0.1 1.093\\
0.2 1.096\\
0.3 1.100\\
0.4 1.105\\
0.5 1.110\\
0.6 1.118\\
0.7 1.129\\
0.8 1.146\\
};
\addlegendentry{$\Nuser = 10$}

\addplot [norm50] 
  table[row sep=crcr]{%
0.0 1.247\\
0.02 1.248\\
0.04 1.250\\
0.06 1.251\\
0.08 1.253\\
0.1 1.255\\
0.14 1.259\\
0.18 1.263\\
0.22 1.267\\
0.26 1.272\\
0.3 1.276\\
0.34 1.282\\
0.38 1.287\\
0.42 1.294\\
0.46 1.301\\
0.5 1.308\\
0.54 1.317\\
0.58 1.326\\
0.6 1.331\\
0.62 1.337\\
0.64 1.343\\
0.66 1.349\\
0.68 1.356\\
0.7 1.364\\
0.72 1.372\\
0.74 1.381\\
0.76 1.391\\
0.78 1.403\\
0.8 1.416\\
0.82 1.430\\
0.84 1.447\\
0.86 1.468\\
0.88 1.493\\
0.9 1.524\\
0.92 1.564\\
0.94 1.622\\
0.96 1.715\\
};
\addlegendentry{$\Nuser = 50$}

\addplot [norm100] 
  table[row sep=crcr]{%
0.0 1.377\\
0.01 1.378\\
0.02 1.379\\
0.03 1.380\\
0.04 1.382\\
0.05 1.383\\
0.09 1.388\\
0.13 1.394\\
0.17 1.400\\
0.21 1.407\\
0.25 1.414\\
0.29 1.421\\
0.33 1.430\\
0.37 1.438\\
0.41 1.448\\
0.45 1.459\\
0.49 1.470\\
0.53 1.483\\
0.57 1.498\\
0.61 1.515\\
0.65 1.534\\
0.69 1.556\\
0.73 1.583\\
0.77 1.615\\
0.8 1.645\\
0.81 1.657\\
0.82 1.669\\
0.83 1.682\\
0.84 1.697\\
0.85 1.713\\
0.86 1.730\\
0.87 1.749\\
0.88 1.770\\
0.89 1.794\\
0.9 1.821\\
0.91 1.852\\
0.92 1.888\\
0.93 1.931\\
0.94 1.984\\
0.95 2.051\\
0.96 2.139\\
0.97 2.266\\
0.98 2.473\\
};
\addlegendentry{$\Nuser = 100$}

\end{axis}
\end{tikzpicture}
%%%%%%%%%%%%%%%%%%%%%%%%%%%%%%%%
%Second Part of Latex include:
%end{figure}
    \caption{Communication per user, normalized by the payload of $\kC\log_2\q \approx 2^{38}$\,bit, versus $\coll/\Nuser$ for $\kC = 10^{10}$, $\qints = 2^{16}$ and $\Ncm=\Nuser$, compared to direct secret sharing, as in SwiftAgg+~\cite{jahani2022swiftagg+}. \vspace{-.2cm}}
    \label{fig:comparison}
\end{figure}

\noindent\textbf{Computational cost.}
\Cref{prop:rate} provides the $\KAHE$ with codes that encode and decode in quasi-linear time.
This does not, however, make encryption and decryption efficient: both require computing $\sV\aM$, which takes $\bigO{\m\,\n}$ operations for unstructured $\aM\in\F_\q^{\m\times\n}$, i.e., about $\m\approx10^8$ operations per ciphertext symbol for the parameters of \Cref{ex:copt_crt}.

To address this, $\aM$ can be sampled from a structured $\DistA$.
As the reductions of \Cref{thm:LPN_to_HintLPN,thm:KAHE_reduction} act only on $\bV$ and $\hV$, never on $\aM$, they hold for any $\DistA$.
For a quasi-cyclic $\aM = [\aM_1 \cdots \aM_{\n/\m}]$, as in ring-LPN~\cite{heyse2012lapin} or HQC~\cite{NISTPQC-R4:HQC22}, computing $\sV\aM$ amounts to $\n/\m$ multiplications in $\F_\q[x]/(x^\m-1)$.
These can be performed in time $\bigO{\n\log\m\log\log\m}$~\cite{cantor1991fast}, so that encryption and decryption are quasi-linear in $\n$.
To avoid attacks exploiting a factorization of $x^\m-1$~\cite{guo2015new}, $\m$ is chosen as a primitive prime~\cite{NISTPQC-R4:HQC22}; the remaining structure gives the best known attacks a speed-up of about $\sqrt{\m}$~\cite{sendrier2011decoding}, which is much smaller than the polynomial factors conservatively ignored in our ISD cost estimate.

\section{Conclusion}\label{sec:conc}
A key- and message-additive homomorphic encryption (KAHE) scheme based on LPN is proposed and used with committee-based key aggregation for post-quantum secure aggregation.
This work addresses challenges of moving from the Euclidean metric~\cite{C:KarPol25,C:BGLRS25} to the Hamming metric.
Decryption amounts to decoding over a $\q$-ary symmetric channel whose crossover probability grows with the number of users.
Leakage of the key sum reduces security to LPN with a noise rate that is quadratic in the KAHE noise rate, and a matching attack shows that this loss is inherent when all but two users collude.
We propose a CRT decomposition to decouple plaintext and LPN modulus, significantly improving communication efficiency.
The communication cost is optimized asymptotically in closed form and outperforms information-theoretically secure aggregation in cross-silo settings starting with moderate collusion.

\section*{Acknowledgment}
The AI assistant Claude Opus 5.5 (Anthropic) was used in preparing this paper.
\Cref{prop:opt} was initially developed with Claude and subsequently checked and revised by the authors.
Claude was also used for language editing and for checking the consistency of the manuscript. The authors take full responsibility for the content of this paper.

\bibliographystyle{IEEEtran}
\bibliography{aux/abbrev3,aux/crypto,aux/refs,aux/refs_pa,aux/refs_pq}

\appendix[Deferred Proofs]\label{app:proofs}
\subsection{Proof of \Cref{thm:KAHE_reduction}}

As in \Cref{def:semantic_secure}, fix messages $\mVu[1],\ldots,\mVu[\Nuser']$ and set $\cVu \define \C.\enc(\mVu)$ for all $\uidx\in[\Nuser']$.
Set $\uVu[1] = \60$ and $\uVu\getrand\F_\q^\n$ for $2 \leq \uidx \leq \hidx$.
For $\hidx \in[\Nuser']$, define the hybrid distribution
\begin{align*}
&\Hyb{\hidx} =
\left(
\aM,
\saggV = \sum\nolimits_{\uidx\in[\Nuser']} \sVu,
\eaggV =\sum\nolimits_{\uidx\in[\Nuser']}\eVu,
\ctVu[1], \ldots,\ctVu[\Nuser']
\right) \\
&\begin{aligned}
\text{with}\quad
\ctVu[1] &= \sum\nolimits_{\uidx\in[\hidx]} \left(\sVu\aM+\eVu - \uVu \right) +\cVu[1],\hspace{-1cm}\\
\ctVu &= \uVu + \cVu &\text{ for } 2\leq\,\, &\uidx \leq \hidx, \\
\ctVu &= \sVu\aM + \eVu + \cVu   &\text{ for } \hidx +1 \leq\,\, &\uidx \leq \Nuser'.
\end{aligned}
\end{align*}
Then, $\Hyb{1}$ is the view of the server as in the left-hand side of \Cref{def:semantic_secure} extended by $\eaggV$, which is available after decoding.
$\Hyb{\Nuser'}$ is a distribution that depends on the individual messages only via their sum $\maggV = \sum_{\uidx\in[\Nuser']} \mVu$ as $\ctVu[1] = \saggV\aM+\eaggV+\C.\enc(\maggV)-\sum_{\uidx\geq2} \ctVu$ and $\ctVu[2],\ldots,\ctVu[\Nuser']$ are uniformly random.
Applying the hybrid argument below to $\mVu[1],\ldots,\mVu[\Nuser']$ and $\mVu[1]',\ldots,\mVu[\Nuser']'$ separately shows that the respective distributions $\Hyb{1}$ are indistinguishable from $\Hyb{\Nuser'}$, which proves \Cref{def:semantic_secure}.

Let $\advA$ distinguish between $\Hyb{\hidx-1}$ and $\Hyb{\hidx}$.
From $\advA$, we construct $\advB$ which breaks $\HintLPN[\p,\,\p]$, i.e., distinguishes  between $(\aM,\,\sV\aM+\eV,\,\hV = \eV+\fV)$ and $(\aM,\,\uV,\,\hV = \eV+\fV)$.
\vspace{0.3cm}
\algrenewcommand\algorithmicfunction{\textbf{Distinguisher}}
\begin{algorithmic}
\Function{$\advB$}{$\aM, \bV, \hV$}
    \vspace{0.15cm}
    \State Set $\cVu = \C.\enc(\mVu)$ for the given messages, $\uidx\in[\Nuser']$
    \State Sample $\eVu \getrand \Bern{\p}^\n$ for $\uidx \in [\Nuser']\setminus\{1,\hidx\}$
    \State Sample $\tilde{\sV}_1 \getrand \F_\q^\m$ and $\sVu \getrand \F_\q^\m$ for $\uidx \in [\Nuser']\setminus\{1,\hidx\}$

    \State Sample ${\6u}_\uidx \getrand \F_\q^\n$ for $1 \leq \uidx \leq \hidx-1$
    \vspace{0.15cm}
    \State 
    Set $\ctVu[1] = \tilde{\sV}_1 \aM 
    +\sum_{\uidx = 2}^{\hidx-1} \left(\sVu \aM + \eVu -\uVu\right)+ \hV - \bV + \cVu[1]$

    \State Set $\ctVu = \uVu + \cVu$ for $2 \leq \uidx \leq \hidx -1$

    \State Set $\ctVu[\hidx] = \bV + \cVu[\hidx]$

    \State Set $\ctVu = \sVu \aM + \eVu + \cVu$ for $\hidx+1\leq \uidx \leq \Nuser'$

    \State Set $\saggV = \tilde{\sV}_1 + \sum_{\uidx\in[\Nuser']\setminus\{1,\hidx\}} \sVu$, $\eaggV = \sum_{\uidx\in[\Nuser']\setminus\{1,\hidx\}} \eVu + \hV$
    \vspace{0.15cm}
    \State \Return $\advA(\aM,\,\saggV,\,\eaggV,\,\ctVu[1],\ldots, \ctVu[\Nuser'])$
\EndFunction
\end{algorithmic}
\vspace{0.3cm}

\noindent\textbf{Case $(\aM, \sV \aM + \eV, \hV)$:}
Given $(\aM, \sV \aM + \eV, \hV=\eV+\fV)$, $\advB$ constructs $\Hyb{\hidx-1}$, using $\sV$, $\eV$ and $\fV$ as $\sVu[\hidx]$, $\eVu[\hidx]$ and $\eVu[1]$, as well as $\sVu[1] \define \tilde{\sV}_1-\sV$.
This gives
\begin{align*}
\6y_1
&= (\tilde{\sV}_1-\sV)\aM +\sum\nolimits_{\uidx = 2}^{\hidx-1} \left(\sVu \aM + \eVu -\uVu\right)+ \hV - \eV + \cVu[1]\\
&= \sum\nolimits_{\uidx\in[\hidx-1]} \left(\sVu\aM+\eVu - \uVu \right) +\cVu[1]  \\
\6y_\hidx
&= \sV\aM+\eV+\cVu[\hidx]= \sVu[\hidx]\aM+\eVu[\hidx]+\cVu[\hidx],
\end{align*}
and consistent $\saggV$, $\eaggV$.

\noindent\textbf{Case $(\aM,\uV,\hV)$:}
$\advB$ constructs $\Hyb{\hidx}$ using $\uV$, $\hV$ and $\tilde{\sV}_1$ as $\uVu[\hidx]$, $\eVu[1]+\eVu[\hidx]$ and $\sVu[1]+\sVu[\hidx]$,
which gives
\begin{align*}
\6y_1
&= \tilde{\sV}_1 \6A
    +\sum\nolimits_{\uidx = 2}^{\hidx-1} \left(\6s_\uidx \6A + \6e_\uidx -\6u_\uidx\right)+ \6h - \6u + {\cw}_1 \\
&= \sum\nolimits_{\uidx\in[\hidx]} \left(\6s_\uidx\6A+\6e_\uidx - \6u_\uidx \right) +\cw_{1},\\
\6y_\hidx &= \6u +\cw_\hidx =  \6u_\hidx +\cw_\hidx.
\end{align*}
By the hybrid argument over $\hidx$, the advantage of $\advA$ is bounded by the advantage of $\advB$ against $\HintLPN[\p,\,\p]$.
\Cref{thm:KAHE_reduction} follows from \Cref{thm:LPN_to_HintLPN}. \qed

\subsection{Proof of \Cref{prop:KAHE_solver}}
Pick $\mVu[1],\ldots,\mVu[\Nuser']$ and $\mVu[1]',\ldots,\mVu[\Nuser']'$ so that both have sum $\maggV$ while $\C.\enc(\mVu[1]-\mVu[1]')$ has Hamming weight at least $\kC$, existence is guaranteed for any $\C$.
Let the server be given $(\aM,\,\saggV,\,\ctVu[1],\ldots,\ctVu[\Nuser'])$, encrypting one of the two; below, we assume that $\mVu[1],\ldots,\mVu[\Nuser']$ was encrypted.
The server computes the error sum as
\[
\eaggV = \sum_{\uidx\in[\Nuser']} \ctVu - \saggV\aM - \C.\enc(\maggV) = \sum_{\uidx\in[\Nuser']}\eVu
\]
and sets $\Iset\define\{i\in[\n]\colon\eaggV_i=0\}$.
As $\eaggV\sim\Bern{\ppile{\Nuser'}}^\n$, $\size{\Iset} \geq \n'$ for $\n' = \lfloor(1-\ppile{\Nuser'}-\varepsilon)\n\rfloor$ except with probability $\exp(-2\varepsilon^2\n)$ by Hoeffding's inequality.
Thus, we may conservatively assume $\size{\Iset} = \n'$ in the following.
Conditioned on $\eaggS_\vidx=0$, $\eS_{1,\vidx}\sim\Bern{\plpn'}$ with $\plpn'=\postpp{\p}{\ppile{\Nuser'-1}}$ by \Cref{eq:postrate}.
The server forms
\[
\bV = \ctVu[1] - \C.\enc(\mVu[1]),
\quad
\bV' = \ctVu[1] - \C.\enc(\mVu[1]'),
\]
restricts $\aM,\bV,\bV'$ to $\Iset$ and applies the solver for $\LPN[\plpn']$ to both $(\aM_\Iset,\,\bV_\Iset)$ and $(\aM_\Iset,\,\bV'_\Iset)$.

By the assumption of \Cref{prop:KAHE_solver}, the solver applied to $(\aM_\Iset,\,\bV_\Iset)$ yields $\sVu[1]$.
The server computes the residual
\[
\eVu[1] = \ctVu[1] - \C.\enc(\mVu[1]) - \sVu[1]\aM,
\]
which has weight at most $(\p+\varepsilon)\n$ except with probability $\exp(-2\varepsilon^2\n)$ by Hoeffding's inequality.

For $(\aM_\Iset,\,\bV'_\Iset)$, the solver may return $\hat{\sV}_1$.
The corresponding residual is of the form
\[
\hat{\eV}_1
= (\sVu[1]-\hat{\sV}_1)\aM + \eV_1 + \C.\enc(\mVu[1]-\mVu[1]').
\]
The case of $\hat{\sV}_1 = \sVu[1]$ is covered for $\kC > 2(\p+\varepsilon)\n$, which guarantees a residual of weight at least $(\p+\varepsilon)\n$ and is satisfied for parameters of interest.
For $\hat{\sV}_1 \neq \sVu[1]$, a union bound shows that no residual with weight at most $(\p+\varepsilon)\n$ exists, except with probability $\q^{-\n(1-\tfrac{\m}{\n}-\qentropy{\p+\varepsilon})}$, which is negligible for parameters of interest.

Combining both cases, the server therefore outputs the message vector whose residual is lighter, deciding correctly except with probability at most $3\exp(-2\varepsilon^2\n) + \q^{-\n(1-\tfrac{\m}{\n}-\qentropy{\p+\varepsilon})}$.
If instead $\mVu[1]',\ldots,\mVu[\Nuser']'$ was encrypted, the same argument applies with the roles of $\bV$ and $\bV'$ exchanged.
Hence, the two distributions of \Cref{def:semantic_secure} are distinguished except with negligible probability. \qed

\subsection{Proof of \Cref{prop:opt}}

\noindent\textbf{Preliminaries.} Write $P \define \ppile{\Nuser}$, which is increasing in $\p$ by \Cref{lem:piling}.
As $\log_2\q$ is constant, we may minimize
\[
\bar{C}(P) \define a\,\m(P) + \frac{\kC}{R(P)}
\]
over $P\in(0,\tfrac{\q-1}{\q})$.
\Cref{lem:piling} gives $1-\tfrac{\q}{\q-1}\p = \bigl(1-\tfrac{\q}{\q-1}P\bigr)^{1/\Nuser}$, so
\begin{equation}\label{eq:p}
\p(P) = \Nuser^{-1} P (1+\bigO{P}),
\end{equation}
and, together with \Cref{eq:postrate},
\[
a\,\m(P) = \frac{A}{P^2}\bigl(1+\bigO{P}\bigr), \qquad A \define a\lambda\ln 2\,(\q-1)\Nuser^2.
\]
Since the code rate is $R(P) = 1-\qentropy{P}$,
\begin{equation}\label{eq:C_KAHE}
\frac{k}{R(P)} = \kC (1+\qentropy{P}+\bigO{\qentropy{P}^2}).
\end{equation}
The corresponding derivatives are
\begin{equation}\label{eq:part_derivatives}
\frac{\dd}{\dd P} a \m(P) = -\frac{2A}{P^3}(1+\bigO{P}),
\quad\
\frac{\dd}{\dd P} \frac{\kC}{R(P)} =  \frac{\kC\,\qentropy[']{P}}{R(P)^2},
\end{equation}
where the derivative of the $\q$-ary entropy $\qentropy[']{P}$ satisfies
\begin{align}
\qentropy[']{P} &= \log_\q\tfrac{(\q-1)(1-P)}{P}, \label{eq:Hqderivative}\\
\qentropy{P} &= P\,\qentropy[']{P} - \log_\q(1-P). \label{eq:Hq_vs_derivative}
\end{align}

\noindent\textbf{Bound on $P^\star$.}
To develop a lower and upper bound on the minimizer $P^\star$, we utilize $\bar{C}(P^\star) \leq \bar{C}(P_0)$, where we choose $P_0 = \kC^{-1/3}$.
There,
\[
\qentropy{P_0} = \bigO{\kC^{-1/3}\log\kC},
\quad
a\,\m(P_0) = \bigO{\kC^{2/3}},
\]
and therefore $\bar{C}(P_0) = \kC + \bigO{\kC^{2/3}\log\kC}$.
Using $a\,\m(P)\geq 0$, $\qentropy{P}\geq P$ and $(1-x)^{-1} \geq (1+x)$,
we get
\begin{align*}
\kC + \bigO{\kC^{2/3}\log\kC}
&= \bar{C}(P_0)
\geq \bar{C}(P^\star)
\geq \frac{\kC}{1-\qentropy{P^\star}} \\
&\geq \kC(1+\qentropy{P^\star})
\geq \kC(1+{P^\star}),
\end{align*}
which implies $P^\star \leq \bigO{\kC^{-1/3}\log\kC}$; in particular $P^\star\to0$.
Hence $a\,\m(P^\star)\geq \frac{A}{2(P^\star)^2}$ for $\kC$ large enough, and, as $R\leq1$,
\begin{align*}
\kC + \bigO{\kC^{2/3}\log\kC}
&= \bar{C}(P_0)
\geq \bar{C}(P^\star) \\
&\geq a\,\m(P^\star) +\kC  \geq \kC + \frac{A}{2(P^\star)^2},
\end{align*}
so that $P^\star \geq \Omega(\kC^{-1/3}\log^{-1/2}\kC)$.
The resulting window $\Omega(\kC^{-1/3}\log^{-1/2}\kC) \leq P^\star \leq \bigO{\kC^{-1/3}\log\kC}$ is not tight enough to give $P^\star$ directly, but taking logarithms gives $\log_\q P^\star = -\frac{1}{3}\,\log_\q(\kC)\, (1+\smallO{1})$, and, by \Cref{eq:Hqderivative},
\begin{equation}\label{eq:Hq_as_k}
\qentropy[']{P^\star} = \frac{1}{3}\,\log_\q(\kC)\cdot (1+\smallO{1}).
\end{equation}

\noindent\textbf{Finding $P^\star$.}
$P^\star$ satisfies the stationarity condition $\tfrac{\dd}{\dd P} \bar{C}(P) = 0$.
Via \Cref{eq:part_derivatives} and $R(P^\star) = 1-\smallO{1}$, this constraint gives
\begin{equation}\label{eq:balance_app}
\frac{2\,A}{(P^\star)^{3}} = \kC\qentropy[']{P^\star} (1+\smallO{1}).
\end{equation}
Since $\kC P^3\qentropy[']{P}/R(P)^2$ is increasing on the window, the solution is unique, and inserting \Cref{eq:Hq_as_k} gives
\[
P^\star = \left(\frac{6 A \log\q}{\kC\log \kC}\right)^{1/3} (1+\smallO{1}),
\]
and \Cref{eq:popt} of \Cref{prop:opt} follows from \Cref{eq:p}.

\noindent\textbf{Determine $\bar{C}(P^\star)$.}
Together, \Cref{eq:Hq_vs_derivative,eq:Hq_as_k} show
\[
\delta \define \qentropy{P^\star} = P^\star\,\qentropy[']{P^\star}(1+\smallO{1}) = \tfrac{1}{3}\,P^\star\log_\q\kC\, (1+\smallO{1})
\]
Then, \Cref{eq:balance_app} gives
\[
a\,\m(P^\star) = \frac{\kC\,P^\star\qentropy[']{P^\star}}{2}\,(1+\smallO{1})
= \tfrac{1}{2} \kC\,\qentropy{P^\star}\,(1+\smallO{1}).
\]
For $\kC/R(P^\star)$, \Cref{eq:C_KAHE} yields
\[
\frac{\kC}{R(P^\star)} = \kC \bigl(1+\qentropy{P^\star}(1+\smallO{1})\bigr).
\]
Combining the two terms gives $\bar{C}(P^\star) = \kC\bigl(1+\tfrac{3}{2}\delta\,(1+\smallO{1})\bigr)$, which is \Cref{eq:copt}.

\balance
\end{document}